\documentclass{article}

\usepackage{arxiv}

\usepackage[utf8]{inputenc} % allow utf-8 input
\usepackage[T1]{fontenc}    % use 8-bit T1 fonts
\usepackage{hyperref}       % hyperlinks
\usepackage{url}            % simple URL typesetting
\usepackage{booktabs}       % professional-quality tables
\usepackage{amsmath}
\usepackage{amsthm}
\usepackage{amsfonts}       % blackboard math symbols
\usepackage{nicefrac}       % compact symbols for 1/2, etc.
\usepackage{microtype}      % microtypography
\usepackage[shortlabels]{enumitem}

\hypersetup{
    colorlinks,
    linkcolor={blue!70!black},
    citecolor={blue!70!black},
    urlcolor={black}
}

\usepackage{xparse}
\usepackage{todonotes}
\usepackage{svn-multi}
\svnid{$Id: arxiv.tex 43256 2020-02-25 07:56:02Z ti5cw $}

\usepackage[export]{adjustbox}% http://ctan.org/pkg/adjustbox for includegraphics

\usepackage[ruled,vlined]{algorithm2e}
\SetKwFor{round}{\hspace*{-0.78ex}}{}{end}
\SetKwFunction{randomColor}{randomColor}
\SetKw{KwRet}{return}
\SetKw{Return}{return}
\SetKwInOut{Input}{input}
\graphicspath{{./figures/}}

\newcommand{\abs}[1]{{\left\lvert #1 \right\rvert}}

\newcommand{\suchthat}{\;\ifnum\currentgrouptype=16 \middle\fi|\;}

\RequirePackage{mleftright} % remove additional space from \left \right
\RenewDocumentCommand{\Pr}{ s o o m }{%
  \IfBooleanTF{#1}{%
    \IfNoValueTF{#2}{%
      \specrandop{Pr}[#4]%
    }{%
      \IfNoValueTF{#3}{%
        \specrandop{Pr}\mathopen{#2[}#4\mathclose{#2]}%
      }{%
        \specrandop{Pr}_{#3}\mathopen{#2[}#4\mathclose{#2]}%
      }
    }%
  }{%
    \IfNoValueTF{#2}{%
      \specrandop{Pr}\mleft[#4\mright]%
    }{%
      \specrandop{Pr}_{#2}\mleft[#4\mright]%
    }
  }%
}

\newtheorem{theorem}{Theorem}

\newtheorem{definition}[theorem]{Definition}
\newtheorem{lemma}[theorem]{Lemma}
\newtheorem{corollary}[theorem]{Corollary}

\newcommand\G{{\cal G}}
\newcommand\F{{\cal F}}

\title{Analysis of Amnesiac Flooding}
\author{%
  \href{https://orcid.org/0000-0001-9964-8816}{%
    \includegraphics[height=0.8em]{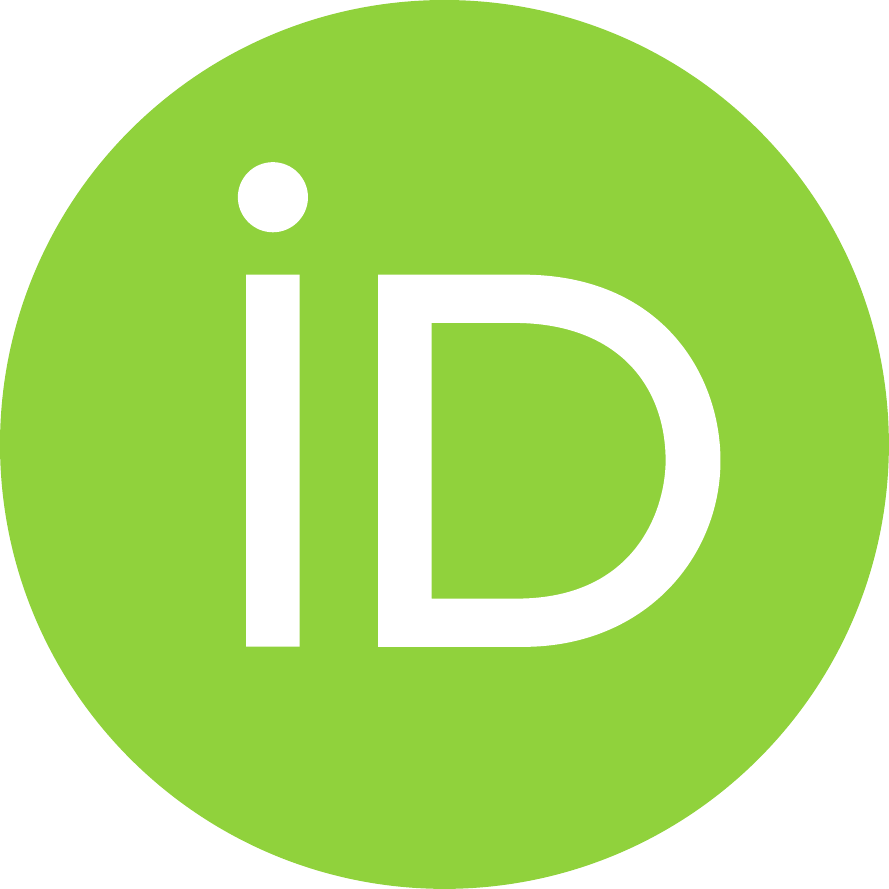}%
    \hspace{1mm}%
    Volker Turau%
  } \\
  Institute of Telematics \\
  Hamburg University of Technology \\
  21073 Hamburg, Germany \\
  \texttt{turau@tuhh.de}
}

\begin{document}

\newcommand{\FL}{\textit{Flood}\xspace}

\newcommand{\fontmathtext}[1]{\mathsf{#1}}%
\DeclareRobustCommand{\alorithm}[1]{{\ensuremath{{\cal A}_\fontmathtext{#1}}}\xspace}
\DeclareRobustCommand{\SynFl}{\alorithm{AF}}

\maketitle

\begin{abstract}
  The broadcast operation in distributed systems is used to spread
  information located at some nodes to all other nodes. This operation
  is often realized by flooding, where the source nodes send a message
  containing the information to all neighbors. Each node receiving the
  message for the first time forwards it to all other neighbors. A
  stateless variant of flooding for synchronous systems is called {\em
    amnesiac flooding}. In this case, every time a node receives a
  message, it forwards it to those neighbors, from which it did not
  receive the message in the current round. The algorithm is oblivious
  and therefore scales very well. Stateless protocols are advantageous
  in high volume applications, increasing performance by removing the
  load caused by retention of session information and by providing
  crash tolerance. In this paper we analyze the termination time of
  amnesiac flooding. We define the {\em $(k,c)$-flooding problem}. It
  asks whether there exists a set $S$ of size $k$, such that amnesiac
  flooding when started concurrently by all nodes of $S$ terminates in
  at most $c$ rounds. We prove that this problem is NP-complete. We
  provide sharp upper and lower bounds for the time complexity of
  amnesiac flooding and reveal a discrepancy between bipartite and
  non-bipartite graphs. All results are based on the insight, that for
  every non-bipartite graph there exists a bipartite graph such that
  the execution of amnesiac flooding on both graphs is strongly
  correlated. This construction considerably simplifies existing
  proofs for amnesiac flooding and allows to analyze the
  $(k,c)$-flooding problem.
\end{abstract}

\keywords{Distributed Algorithms, Flooding}

\section{Introduction}
The most basic algorithm to disseminate information in a distributed
system is deterministic flooding. The originator of the information
sends a message with the information to all neighbors and whenever a
node receives this message for the first time, it sends it to all its
neighbors in the communication graph $G$. This algorithm uses
$2\abs{E}$ messages and terminates in $\epsilon_G(v_0)+1$ rounds,
where $v_0$ is the originating node and $\epsilon_G(v_0)$ is the
maximal distance of $v_0$ to any other node.  The
flooding algorithm requires each node to maintain for each message a
marker that the message has been forwarded. This requires storage per
node proportional to the number of disseminated messages. Another
issue is how long these markers are kept. Thus, stateful algorithms
such as the classic flooding algorithm do not scale well. Stateful
algorithms are therefore unsuitable for resource-constrained devices
as those used in the Internet of Things.

% These bounds hold in the synchronous and the asynchronous case
% \cite{Peleg:2000}.

In {\em stateless protocols} each message travels on it's own without
reference to any previous message. Stateless protocols are very
popular in client-server applications because of their high degree of
scalability. They simplify the design of the server and require less
resources because servers do not need to keep track of session
information or a status about each communicating partner for multiple
requests. In addition they provide fault tolerance after node crashes.
Despite their significance for practical application, stateless
protocols have only received limited attention in theoretical
research.

Recently Hussak and Trehan proposed a stateless information
dissemination algorithm for synchronous systems called {\em amnesiac
  flooding} \cite{Hussak:2019,Hussak:2020}. In this algorithm a node after
receiving a message, forwards it to those neighbors from which it did
not receive the message in the current round. Obviously, this variant of
flooding is stateless and avoids the above mentioned storage issues.
It is not obvious that amnesiac flooding terminates since a node can
potentially forward the same message several times. Hussak et al.\
analyzed the termination time of amnesiac flooding with a single
originating node $v_0$ \cite{Hussak:2019,Hussak:2020}. They show that
synchronous amnesiac flooding terminates on any graph. For bipartite
graphs amnesiac flooding terminates after $\epsilon_G(v_0)$ rounds,
i.e., the same number of rounds as the marker based algorithm. In the
non-bipartite case amnesiac flooding requires at least
$\epsilon_G(v_0)+1$ and at most $\epsilon_G(v_0) + Diam(G)+1$ rounds,
where $Diam(G)$ denotes the diameter of $G$. The proof of this result
in \cite{Hussak:2020} is rather technical and does not give much
insight into the problem.

Hussak et al.\ stated the following as an open problem: What happens
when multiple nodes $S$ start the flooding process with the same
message? In this paper we give a detailed answer to this problem. The
problem of multi-source flooding is motivated by disaster monitoring:
A dense distributed monitoring system monitors a geographical region.
When several sensors detect an event, they flood this information into
the network.

The contribution of this paper is twofold. First of all we prove that
for every non-bipartite graph and every set $S$ there exists a
bipartite graph such that the execution of amnesiac flooding on both
graphs is strongly correlated and termination times coincide. The
auxiliary graph captures on an intuitive level, what happens during
the flooding process. This graph considerably simplifies the proof of
\cite{Hussak:2020}. Besides leading to better bounds it also allows to
determine starting nodes for amnesiac flooding with minimal
termination time. We give upper and lower bounds for the time
complexity of amnesiac flooding. Secondly we define the {\em
  $(k,c)$-flooding problem}. Let $G=(V,E)$ be a connected graph and
$c,k$ positive integers with $k\le|V|$. The problem is to decide
whether there exists a subset $S$ of $V$ of size $k$, such that
amnesiac flooding when started concurrently by all nodes of $S$
terminates in at most $c$ rounds. We prove this problem is NP-complete
and illustrate the relationship of this problem to the classical {\em
  k-center} problem \cite{Kariv:1979,Hsu:1979}.

After introducing our notation and reviewing the state of the art we
present in Section~\ref{sec:basic-properties} our implementation
\SynFl of amnesiac flooding. In the next section we summarize our main
results. In Section~\ref{sec:s1} we present the construction of the
auxiliary graph $\G$ for the case $|S|=1$ and demonstrate that this
graph immediately proves the results of
\cite{Hussak:2019,Hussak:2020}. In the following section we reduce the
case $|S|>1$ to the case $|S|=1$ and prove the main theorems. We
conclude the paper with some open problems.

\section{Notation}
\label{sec:notation}
In the following $G(V,E)$ is always a finite, connected, undirected,
and unweighted graph with $n=|V|$ and $m=|E|$. The minimal node degree
of $G$ is denoted by $\delta$. For $u,v\in V$ denote by $d_G(v,u)$ the
{\em distance} in $G$ between $v$ and $u$, i.e., the number of edges
of a shortest path between $v$ and $u$. For $U\subseteq V$ and
$v\in V$ let $d_G(v,U)= \min\{d_G(v,u)\suchthat u\in U\}$ and
$d_G(U)= \max \{d_G(v,U) \suchthat v\in V\}$. Furthermore, $G[U]$
denotes the graph induced by $U$. For $v\in V$ denote by
$\epsilon_G(v)$ the {\em eccentricity} of $v$ in $G$, i.e., the
greatest distance between $v$ and any other node in $G$, i.e.,
$\epsilon_G(v)=d_G(\{v\})$. In a few cases we consider disconnected
graphs. In this case we define $\epsilon_G(v)=\epsilon_U(v)$, where
$U$ is the connected component of $G$ containing $v$. The {\em radius}
$Rad(G)$ (resp.\ {\em diameter} $Diam(G)$) of $G$ is the minimum
(resp.\ maximum) eccentricity of any node of $G$. A {\em central} node
in $G$ is a node $v$ with $\epsilon_G(v)=Rad(G)$. An edge $(u,w)\in E$
is called a {\em cross edge} with respect to a node $v_0$ if
$d_G(v_0,u) = d_G(v_0,w)$. Any edge of $G$ that is not a cross edge
with respect to $v_0$ is called a {\em forward edge} for $v_0$.

Let $n\ge k\ge 1$ be an integer. We call
$r_k(G) = \min \{d_G(U)\suchthat \abs{U}=k\}$ the {\em k-radius} of
$G$. Thus, $r_1(G)=Rad(G)$. A subset $U\subseteq V$ with $\abs{U}=k$
and $r_k(G)=d_G(U)$ is called a {\em $k$-center} of $G$. Similarly we
call
$r^{ni}_k(G)=\min \{d_G(U)\suchthat \abs{U}=k \mbox{ and } G[U] \mbox{
  contains no isolated node}\}$ the {\em non-isolated k-radius} of
$G$. Clearly $r_k(G) \le r^{ni}_k(G)$. A subset $U\subseteq V$ with
$\abs{U}=k$ such that $G[U]$ has no isolated node and
$r^{ni}_k(G)=d_G(U)$ is called a {\em non-isolated $k$-center} of $G$.
A {\em total dominating set} of a connected graph $G$ is a set $S$ of
nodes of $G$ such that every node is adjacent to a node in $S$. The
{\em total domination number} of $G$, denoted by $\gamma_t(G)$, is the
minimum cardinality of a total dominating set of $G$. Note that a
non-isolated $k$-center with radius $1$ is a total dominating set.
Thus, $\gamma_t(G)\le k$ if and only if $r^{ni}_k(G)=1$.

Throughout the paper we consider a synchronous distributed system.
This means that algorithms are executed in rounds of fixed lengths and
all messages sent by all nodes in a particular round are received and
processed in the next round. Furthermore, no messages are lost or
corrupted. An algorithm terminates, when all messages sent are
received. For a discussion of asynchronous amnesiac flooding we refer
to \cite{Hussak:2020}.

% The independence number $\alpha(G)$ of a graph $G$, is the cardinality
% of the largest independent vertex set of $G$.
%\[\argminA_x f(x) = \{x \mid f(x) = \min_{x'} f(x')\} \]

\section{State of the Art}
Different facets of stateless programming received a lot of attention
in recent years: MapReduce framework, monads in functional
programming, and reentrant code \cite{Dolev:2014}. Stateless protocols
are very popular in client-server applications because of their high
degree of scalability. They simplify the design of the server and
require less resources because servers do not need to keep track of
session information or a status about each communicating partner for
multiple requests. According to Awerbuch et al.\ statelessness implies
various desirable properties of distributed algorithms, such as:
asynchronous updates and self-stabilization \cite{Awerbuch:2008}. 

Broadcast in computer networks has been the subject of extensive
research. The survey paper \cite{Hedetniemi:1988} covers early work.
The classic flooding algorithm, where each node that receives the
message for the first time forwards it to all other neighbors,
requires in the worst case $Diam(G)$ rounds until all nodes have
received the message and uses $O(m)$ messages \cite{Peleg:2000}. This
result holds both in the synchronous and the asynchronous model.

% The number of messages is reduced to $O(n)$ if flooding is performed
% via the edges of a spanning tree only.

The classic flooding algorithm is a stateful algorithm. Each node
needs to maintain for each message a marker that the message has been
forwarded. This requires storage per node proportional to the number
of disseminated messages. A stateless version of flooding was proposed
by Hussak and Trehan \cite{Hussak:2019}. Their algorithm -- called
amnesiac flooding -- forwards every received message to those
neighbors from which it did not receive the message in the current
round. Amnesiac flooding has a low memory requirement since
markers are only kept for one round. Note, that nodes may forward a
message more than once. In synchronous networks amnesiac flooding when
started by a node $v_0$ terminates after at most
$\epsilon_G(v_0) + Diam(G) + 1$ rounds \cite{Hussak:2020}. The proof
is based on an analysis of the forwarding process on a round by round
basis, whereas our analysis is based on an auxiliary
graph. We believe that our approach opens more possibilities for more
general problems related to amnesiac flooding. To the best of our
knowledge, the problem of simultaneously starting the flooding process
from many nodes has not been covered in the literature.

The $k$-center problem received a lot of attention since it was first
proposed \cite{Kariv:1979}. The task is to find a $k$-center of a
graph. The problem and many variants of it including some
approximations are known to be NP-hard \cite{Calik:2015,Hsu:1979}.

A problem related to broadcast is rumor spreading that describes the
dissemination of information in large and complex networks through
pairwise interactions. A simple model for rumor spreading is to assume
that in each round, each node that knows the rumor, forwards it to a
randomly chosen neighbor. For many network topologies, this strategy
is a very efficient way to spread a rumor. With high probability the
rumor is received by all vertices in time $\Theta(\log n)$, if the
graph is a complete graph or a hypercube
\cite{Frieze:1985,Feige:1990}. New results about rumor spreading can
be found in \cite{Mocquard:2020}.

\section[Amnesiac Flooding: Algorithm A\_AF]{Amnesiac Flooding:  Algorithm \SynFl}
 \label{sec:basic-properties}
 The goal of amnesiac flooding is to distribute a message -- initially
 stored at all nodes of a set $S$ -- to all nodes of the network. In
 the first round each node of $S$ sends the message to all its
 neighbors. In each of the following rounds each node that receives at
 least one message forwards the message to all of its neighbors from
 which it did not receive the message in this round. The algorithm
 terminates, when no more messages are sent. Algorithm~\ref{alg:fl}
 shows a formal definition of algorithm \SynFl. The code shows the
 handling of a single message $m$. If several messages are
 disseminated concurrently, each of them requires its own set $M$.

 \SetProgSty{}
 \begin{algorithm}%[H]%\footnotesize
  \Input{A graph $G=(V,E)$, a subset $S$ of $V$, and a message $m$.}
  \BlankLine
  In round $1$ each node $v\in S$ sends message $m$ to each neighbor
  in $G$\;
  %\Indp
   \round{Each node $v$ executes in every round $i>1$}{
   $M := N(v)$\;
   \ForEach{receive$(w,m)$}{$M := M \setminus \{w\}$}
   \If{$M \not=N(v)$} {
   \lForAll{$u\in M$}{send($u,m$)}}
 }
\caption{Algorithm \SynFl distributes a message in the
  graph $G$}\label{alg:fl}
\end{algorithm}

To illustrate the flow of messages of algorithm \SynFl we consider a
graph with four nodes as depicted in Fig.~\ref{fig:examples} (nodes
in $S$ are depicted in black). The top two rows show the flow of
messages for two different choices for $S$ with $|S|=1$. In the first
case \SynFl terminates after three rounds and in the second case after
five. The last row of this figure shows an example with $|S|=2$. In
this case the algorithm also terminates after three rounds.

\begin{figure}[h]%
    \includegraphics[valign=t,scale=0.99]{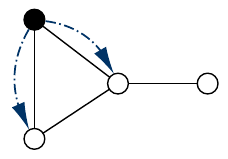}
    \hspace*{0.4ex}
    \includegraphics[valign=t,scale=0.99]{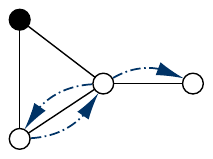}
    \includegraphics[valign=t,scale=0.99]{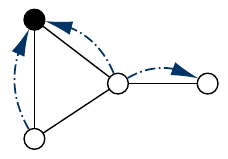}

    \vspace*{0.8ex}

    \hspace*{0.4ex}
    \includegraphics[valign=t,scale=0.99]{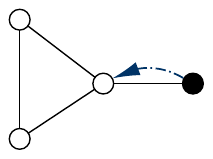}
    \hspace*{0.4ex}
    \includegraphics[valign=t,scale=0.99]{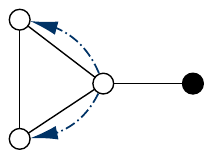}
    \includegraphics[valign=t,scale=0.99]{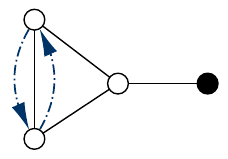}
    \hspace*{0.4ex}
    \includegraphics[valign=t,scale=0.99]{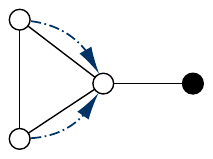}
    \hspace*{0.4ex}
    \includegraphics[valign=t,scale=0.99]{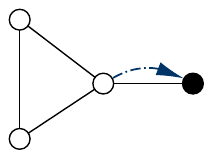}

    \vspace*{0.8ex}

    \includegraphics[valign=t,scale=0.99]{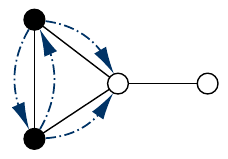}
    \hspace*{0.4ex}
    \includegraphics[valign=t,scale=0.99]{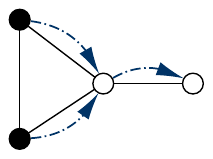}
    \hspace*{0.4ex}
    \includegraphics[valign=t,scale=0.99]{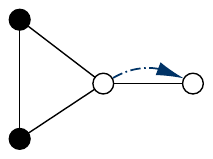}

  \caption{Three executions (one per row) of algorithm \SynFl for different choices of $S$.\label{fig:examples}}
\end{figure}

These examples demonstrate that the termination time of \SynFl highly
depends on $S$. This is captured by the following definition.

\begin{definition}
  For $S\subseteq V$ denote by $\FL_G(S)$ the number of rounds
  algorithm \SynFl requires to terminate when started by all nodes in
  $S$. For $1\le k \le n$ define
  \[\FL_k(G)=\min \{\FL_G(S) \suchthat S \subseteq  V \text{ with }
    |S|=k\}.\]
\end{definition}

Obviously, $\FL_n(G)=1$ for any graph $G$. For a complete graph $K_n$
with $n>2$ we have $\FL_i(K_n)=2$ for $1<i<n$ and $\FL_1(K_n)=3$. For
a cycle graph $C_n$ we have $\FL_k(C_n)=\lceil n/k\rceil$ if
$n\equiv 1 (2)$ and otherwise
\[\FL_k(C_n) =
  \begin{cases}
    \lceil n/(2k)\rceil &\quad\text{if } k\le n/2\\
    1 &\quad\text{if } k = n \\
    2 &\quad\text{otherwise.} \\
  \end{cases}\]

An instance of the {\em $(k,c)$-flooding problem} is defined as
follows. Let $G=(V,E)$ be a connected graph and $c,k$ positive
integers with $k\le|V|$. Is there a subset $S$ of $V$ of size $k$ such
that $\FL_G(S)\le c$? One might think that the $k$-center of a graph
is a good choice for $S$ to minimize $\FL_G(S)$. This is not always
the case. For the graph depicted in Fig.~\ref{fig:anticenter} with
$n\equiv 0 (2)$ the $1$-center consists of the node with distance
$n/2 -1$ to the rightmost node. Algorithm \SynFl started in this
central node terminates after $3n/2-2$ rounds. Whereas the minimal
value of $n-1$ rounds is independently of $n$ achieved for each of the
two leftmost nodes, i.e.\ $\FL_1(G)=n-1$. Note that $Rad(G)=n/2-1$.

\begin{figure}[h]%
  \hfill
    \includegraphics[valign=t]{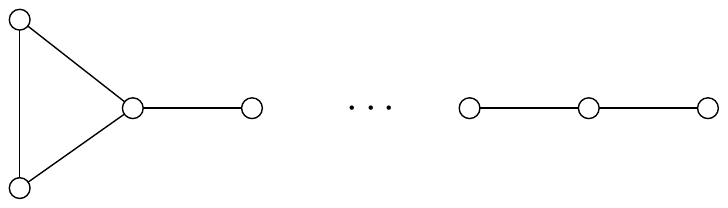}
  \hfill\null
  \caption{A graph $G$ with $\FL_1(G)=n-1$ and $\FL(v_0)=3n/2-2$ where
    $v_0$ is a central node.\label{fig:anticenter}}
\end{figure}

Whereas the $k$-radius monotonically decreases with increasing $k$ for
a fixed graph, this is not generally true for the value of $\FL_k(G)$,
e.g., for $n\equiv 0(2)$ we have $\FL_{n/2}(C_{n})=1$ but
$\FL_{n/2+1}(C_{n})=2$. We will show that monotony holds if $G$ is
non-bipartite.

% The following lemma can be helpful to compute $\FL_{|S|}(G)$ in
% special cases.\todo{Anwendung? Petersen Graph, generalized Petersen graph}

% \begin{lemma}\label{lem:noniso}
%     Let $S\subset V$ such that $G[S]$ has no isolated vertex and
%     $V\setminus S$ is independent. Then $\FL_{|S|}(G)=2$.
%   \end{lemma}
%   \begin{proof}
%     Clearly, $\FL_{|S|}(G)> 1$ since $|S|<n$. In the first round all
%     nodes of $V\setminus S$ receive messages from all their neighbors
%     and thus, they do not send a message in the second round. This is also
%     holds for each node $v\in S$ with $N(v)\subseteq S$. Each node
%     $v\in S$ receives at least one message in the first round since it
%     has at least one neighbor in $S$. Hence, each node $v$ in $S$
%     sends in round two a message to each of its neighbors in
%     $V\setminus S$. Thus, nodes in $V\setminus S$ receive in
%     round 2 messages from all their neighbors and therefore do not
%     send messages in round 2. This yields $\FL_{|S|}(G)= 2$.
%   \end{proof}

\section{The Main Results}%\todo{Is the $k$-flooding problem NP-hard for bipartite graphs?}
% The following two theorems summarize the main results of this paper.
% The proofs are contained in the following sections. The first theorem
% shows that the value of $\FL_k(G)$ significantly depends on whether
% $G$ is bipartite or not. It also provides upper and lower bounds for
% $\FL_k(G)$. The second theorem characterizes graphs with
% $\FL_k(G)\le 2$.

The following theorems summarize the main results of this paper. The
main technical contribution is the auxiliary graph that nicely
captures, also on an intuitive level, what happens during the flooding
process. This bipartite graph leads to a very simple proof that
amnesiac flooding terminates when started with more than one initiator.
The results of \cite{Hussak:2019,Hussak:2020} are easily derived from
the auxiliary graph. Theorem~\ref{theo:fundametal} summarizes the main result.

\begin{theorem}\label{theo:fundametal}
  Let $G=(V,E)$ be a connected graph. For every $S\subseteq V$ there
  is a bipartite graph $\G(S)$ with a node $v^\ast$ such that
  $\FL_G(S)=\FL_{\G(S)}(v^\ast)-1\le d(S,V)+1+Diam(G)$. The number of
  messages Algorithm \SynFl sends is either $\abs{E}$ or $2\abs{E}$.
\end{theorem} 

Whether \SynFl sends $\abs{E}$ or $2\abs{E}$ messages depends on
whether $G^\ast(S)$ is bipartite or not. This graph consists of $G$
and an additional node $v^\ast$ connected to all nodes of $S$.

The bound $\FL_G(S) \le d_G(S,V) +1 + Diam(G)$ coincides with the bound of
\cite{Hussak:2020} for the case $\abs{S}=1$. To see that the bound is
sharp consider a graph $G$ with a single cross edge $(v,w)$ such that
$d_G(v,S)=d_G(w,S)=Diam(G)$ (see Fig.~\ref{fig:sharp}). Apart from
proving upper bounds for $\FL_G(S)$ using other graph parameters, we
focus on lower bounds for $\FL_k(G)$. The auxiliary graph $\G(S)$
allows to derive upper and lower bounds for $\FL_k(G)$ with respect to
$r_k(G)$ and $r_k^{ni}(G)$. Theorem~\ref{theo:mainRes} characterizes
in addition configurations with short termination times. The stated
bounds are sharp as demonstrated by examples.
\begin{figure}[h]%
  \hfill
    \includegraphics[valign=t]{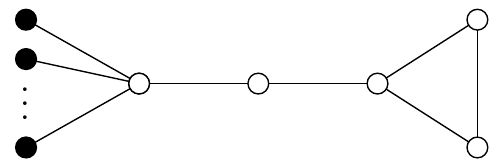}
  \hfill\null
  \caption{A graph with $\FL_G(S) = d_G(S,V) +1 + Diam(G)$ (nodes in S
    are depicted in black).\label{fig:sharp}}
\end{figure}

% In addition we give criteria that
% ensure termination in at most three rounds.

% The graphs $C_n$ with odd $n$ show that for non-bipartite graphs the
% bounds of Theorem~\ref{theo:mainRes} (\ref{theo:mainRes1}) can be
% arbitrarily bad.

\begin{theorem}\label{theo:mainRes}
  Let $G=(V,E)$ be a connected graph. 
  \begin{enumerate}[(1)]
  \item If $k>1$ then $r_k(G)\le \FL_k(G)\le r^{ni}_k(G) +1 \le
    r_{\lfloor k/2\rfloor}(G)+1$.\label{theo:mainRes1}
  \item $\FL_k(G)=1$ if and only if $n=k$ or $G$ is bipartite with
    $|V_1|=k$ or $|V_2|=k$.\label{theo:mainRes2}
  \item $\FL_k(G)\le 2$ if $k\ge 2n/3$ or
    $\delta(G)\ge 3$ and $k\ge n/2$.\label{theo:mainRes3}
  \item $\FL_k(G) \le 3$ if $k\ge n/2$.\label{theo:mainRes4}
  \end{enumerate}
\end{theorem}

Amnesiac flooding takes considerably more time when the graph contains
odd cycles, i.e., is non-bipartite. The following theorem
characterizes non-bipartite for which \SynFl terminates after two
rounds. It also provides a lower bound for the termination time with
respect to the radius of $G$ and $k$ and shows that $\FL_{k}(G)$ is
monotonically decreasing.

\begin{theorem}\label{theo:nonbip}
  Let $G=(V,E)$ be a connected, non-bipartite graph. Then
  \begin{enumerate}[(1)]
  \item $\FL_k(G) \ge r_k(G)+1$.\label{theo:nonbip1}
  \item If $k>1$ then $\FL_k(G)= 2$ if and only if
    $r^{ni}_k(G)=1$.\label{theo:nonbip2}
  \item $\FL_k(G) \ge Rad(G)/k +1/2$ and
    $\FL_{k+1}(G)\le \FL_k(G)$.\label{theo:part3}
  \end{enumerate}
\end{theorem}

The next theorem analyzes the behavior of amnesiac flooding on
bipartite graphs. Somehow surprisingly,
the termination time is at most
two rounds more than the $k$-radius of the graph. For $k\le n/2$ the
difference is at most one round. This is an essential difference to
the case $k=1$ where the termination time always coincides with the
radius. We provide examples showing that the bounds are sharp.

\begin{theorem}\label{theo:special}
  Let $G=(V_1\cup V_2, E)$ be a connected, bipartite graph. If
  $k\ge 1$ then
\begin{enumerate}[(1)]
\item $\FL_k(G)= r_k(G)$ if and only if $G$ has a $k$-center $S$ with
  either  $S\subseteq V_1$ or $S\subseteq V_2$.\label{theo:part2}
\item If $k \le \max(\abs{V_1}, \abs{V_2})$ then
  $\FL_k(G)-r_k(G) \le 1$.\label{theo:special2}
  \item $\FL_k(G)-r_k(G) \le 2$.\label{theo:special3}
  \end{enumerate}
\end{theorem}

Finally we establish that $(k,c)$-flooding is NP-complete.

\begin{theorem}\label{theo:NP}
  The $(k,c)$-flooding problem is NP-complete.
\end{theorem}

\section[The Case |S|=1]{The Case $|S|=1$}
\label{sec:s1}
We begin the introduction of the auxiliary graph for the case
$\abs{S}=1$. With the help of a further graph extension technique we
will in Section~\ref{sec:several-sources} reduce the case $\abs{S}>1$
to this case. Let $S=\{v_0\}$ with $v_0\in V$. The way messages are
forwarded by algorithm \SynFl implies that a message can arrive
multiple times at a node. Clearly, the first time that a message
arrives at a node is along the shortest path from $v_0$ to this node.
The following lemma is easy to prove.
\begin{lemma}\label{lem:levels}
  In round $i>0$ of algorithm \SynFl each node $v$ with
  $d_G(v_0,v)=i-1$ sends a message to all neighbors $u$ with
  $d_G(v_0,u)=i$.
\end{lemma}
Edges that do not belong to a shortest path don't affect the flow of
messages along shortest paths, but they can provoke additional
messages. Such edges are cross edges with respect to a
breadth-first-search starting in $v_0$. Since bipartite graphs have no
cross edges Lemma~\ref{lem:levels} can be strengthened for bipartite
graphs. The proof of the next lemma is by induction on $i$.

\begin{lemma}\label{lem:bilevels}
  Let $G$ be a bipartite graph and $v$ a node with $d_G(v_0,v)=i$. In
  round $i+1$ of algorithm \SynFl node $v$ sends a message to all neighbors
  $u$ with $d_G(v_0,u)=i+1$ and to no other neighbor. In all other
  rounds $v$ does not send a message.
\end{lemma}

\begin{corollary}\label{cor:bipartite}
  If $G$ is bipartite then $\FL_G(v_0)=\epsilon_G(v_0)$ and
  $\FL_1(G)=Rad(G)$.
\end{corollary}

To analyze the behavior of \SynFl for non-bipartite graphs we
introduce the most important concept of this work, the auxiliary graph
$\G(v_0)$.

\subsection[The Auxiliary Graph G(v\_0)]{The Auxiliary Graph $\G(v_0)$}
\label{sec:subst-graph}
For a given graph $G$ and a starting node $v_0$ we define the
auxiliary graph $\G(v_0)$. The executions of \SynFl on $G$ and
$\G(v_0)$ are tightly coupled. Since $\G(v_0)$ is bipartite we can
apply Corollary~\ref{cor:bipartite} to compute $\FL_G(v_0)$.
%\todo{Im  Folgenden muss anstatt $G(v_0)$ oft $B(G,v_0)$ verwendet werden.}

\begin{definition}
  Denote by $\F(v_0)$ the subgraph of $G$ with node set $V$ and all
  edges of $G$ that are not cross edges with respect to $v_0$.
\end{definition}
Obviously $\F(v_0)$ is always bipartite. Fig.~\ref{fig:dag}
demonstrates this definition.

\begin{figure}[h]%
  \hfill
    \includegraphics{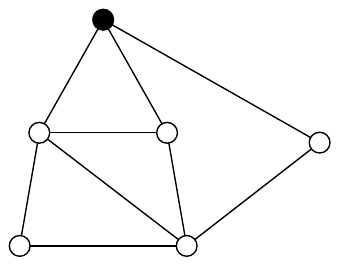}
  \hfill
    \includegraphics{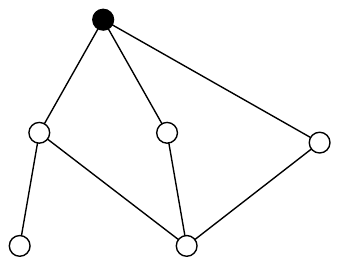}
  \hfill\null
  \caption{On the left the graph $G$ with $v_0$ marked and $\F(v_o)$
    on the right.\label{fig:dag}}
\end{figure}

\begin{definition}
  Denote by $\G(v_0)$ the graph that consists of two copies of
  $\F(v_0)$ with node sets $V$ and $V'$ and additionally for any cross
  edge $(u,w)$ of $G$ the edges $(u,w')$ and $(w,u')$.
\end{definition}
In the following we denote for every $v\in V$ the copy of $v$ in $V'$
by $v'$. $\G(v_0)$ consists of $2\abs{V}$ nodes and $2\abs{E}$ edges.
Every additional edge connects a node from $V$ with a node
from $V'$. Fig.~\ref{fig:trans} demonstrates this construction. For
each $v\in V$ we have
$deg_G(v) = deg_{\G(v_0)}(v) = deg_{\G(v_0)}(v')$. Furthermore,
$d_{\G(v_0)}(v_0,v)=d_{\F(v_0)}(v_0,v)=d_G(v_0,v)$ for $v\in V$.

\begin{figure}[h]%
  \hfill
    \includegraphics[valign=t]{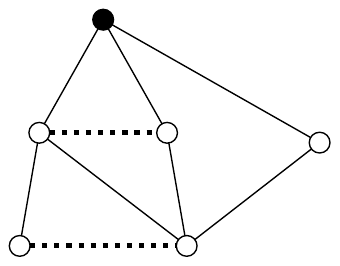}
  \hfill
    \includegraphics[valign=t]{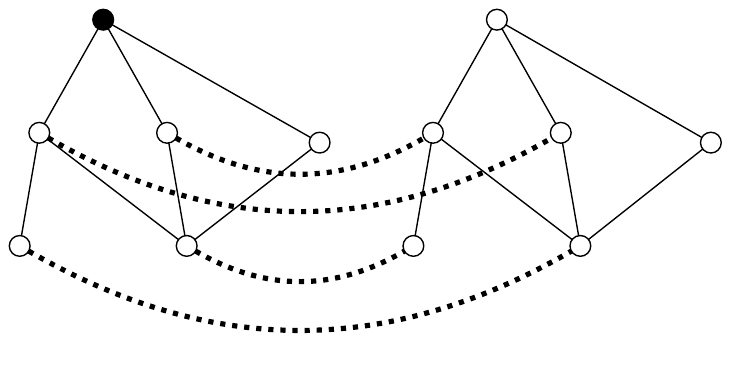}
  \hfill\null
  \caption{The graph $G$ on the left has two cross edges (dotted
    lines), $\G(v_0)$ is shown on the right. \label{fig:trans}}
\end{figure}

\begin{lemma}
   $\G(v_0)$ is bipartite.\label{lem:bipartite}
 \end{lemma}
 \begin{proof}
   $\G(v_0)$ consists of two copies of $\F(v_0)$, with node sets $V$
   and $V'$. $\F(v_0)$ is $2$-colorable.  The cross edges of $G$
   are modified so that one end is in $V$ and the other in $V’$. We
   reverse the coloring of the nodes in $V’$. Since cross edges are
   at the same level in $G$, in the auxiliary graph they connect
   opposite colors. Thus, $\G(v_0)$ is $2$-colorable and,
   hence, bipartite.
 \end{proof}
 
% \begin{proof}
%   It suffices to prove that
%   $d_{\G(v_0)}(v_0,u)\not=d_{\G(v_0)}(v_0,w)$ for every edge $(u,w)$
%   of $\G(v_0)$. There are three cases to consider. If $u,w$ are both
%   nodes of $V$ then any shortest path from $v_0$ to $u$ (resp.\ to
%   $w$) in $\G(v_0)$ is also a shortest path in $\F(v_0)$. Since $\F(v_0)$ is
%   bipartite we have $d_{\G(v_0)}(v_0,w)\not=d_{\G(v_0)}(v_0,u)$.

%   If $u$ is a node of $V$ and $w$ a node of $V'$, i.e., $w = v'$ for
%   some node $v\in V$, then $(u,v)$ is a cross edge of $G$ with respect
%   to $v_0$ and hence $d_{\F(v_0)}(v_0,u)=d_{\F(v_0)}(v_0,v)$. Thus,
%   Lemma~\ref{col:cross} yields
%   \[d_{\G(v_0)}(v_0,w)= d_{\F(v_0)}(v_0,u)+1= d_{\G(v_0)}(v_0,u)+1
%     \not=d_{\G(v_0)}(v_0,u).\]

%   Finally consider the case that $u,w$ are both nodes of $V'$. Thus,
%   $d_{\F(v_0)}(v_0,u)\not=d_{\F(v_0)}(v_0,w)$ since ${\F(v_0)}$ is
%   bipartite and $(u,w)$ is an edge of ${\F(v_0)}$. Let $(u_1,u_2')$
%   (resp.\ $(w_1,w_2')$ be the cross edge on a shortest path from $v_0$
%   to $u$ (resp.\ $w$) in $\G(v_0)$ (Lemma~\ref{col:cross}). Then
%   $d_{\G(v_o)}(v_0,u)= d_{\F(v_0)}(v_0,u_1) + d_{\F(v_0)}(u_2,u)+1 =
%   d_{\F(v_0)}(v_0,u)+1$ and
%   $d_{\G(v_o)}(v_0,w)= d_{\F(v_0)}(v_0,w_1) + d_{\F(v_0)}(w_2,w)+1 =
%   d_{\F(v_0)}(v_0,w) +1$. Hence, again
%   $d_{\G(v_0)}(v_0,w)\not=d_{\G(v_0)}(v_0,u)$.
% \end{proof}

Lemma~\ref{lem:bilevels} and \ref{lem:bipartite} imply the following
result that is frequently used in the following.
\begin{lemma}\label{lem:simple}
  Let $(u,w')$ be an edge of $\G(v_0)$ with $u\in V$ and
  $w'\in V'$. Node $w'$ never sends in $\G(v_0)$ a message to
  $u$ via edge $(u,w')$ but $u$ sends a message via this edge to $w'$.
\end{lemma}

The last lemma has the following simple corollary.

\begin{corollary}\label{col:cross}
  A shortest path in $\G(v_0)$ from $v_0$ to a node $w'\in V'$ uses
  exactly one edge from $V$ to $V'$.
\end{corollary}
% \begin{proof}
%   Assume that a shortest path $P$ from $v_0$ to $w'$ in $\G(v_0)$ uses
%   more than one edge from $V$ to $V'$. Let $(a,b')$ and $(c',d)$ be
%   the first two such edges. The length of $P$ from $v_0$ to $d$ in
%   $\G(v_0)$ is $d_{\F(v_0)}(v_0,a) + 1 + d_{\F(v_0)}(b,c) + 1$. On the
%   other hand $d_{\F(v_0)}(v_0,d)=d_{\F(v_0)}(v_0,a)+d_{\F(v_0)}(b,c)$.
%   Hence $P$ is not a shortest path in $\G(v_0)$. Contradiction.
% \end{proof}

Fig.~\ref{fig:exec} depicts an execution of \SynFl on $\G(v_0)$ for
the graph $G$ shown in Fig.~\ref{fig:trans}. The next lemma shows the
relationship between the execution of \SynFl on $G$ and $\G(v_0)$.

\begin{figure}[h]%
  \hfill
    \includegraphics[valign=t]{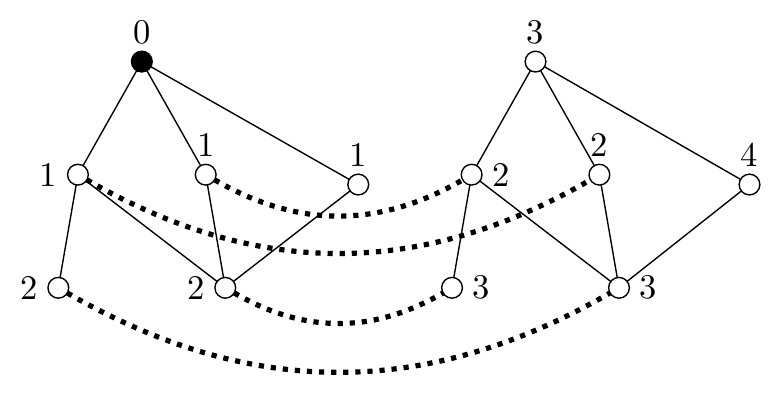}
  \hfill\null
  \caption{The nodes' labels denote the round a message is received by
    the node in $\G(v_0)$.\label{fig:exec}}
\end{figure}

\begin{lemma}\label{lem:equiv}
  Let $v,w\in V$. Node $v$ receives a message from $w$ in round
  $i$ in $G$ if and only if in round $i$ node $v$ receives a message
  from $w$ in $\G(v_0)$, or $v'$ receives a message from $w$ or from
  $w'$ in $\G(v_0)$.
\end{lemma}
\begin{proof}
  The proof is by induction on $i$. The Lemma is true for $i=1$. Note
  that from the three conditions for $\G(v_0)$ only the first can
  occur in round $1$. Let $i>1$.

  First suppose that $w$ sends in $G$ a message to $v$ in round $i$.
  Then in $G$ in round $i-1$ node $w$ received a message from a
  neighbor $z$ with $z\not=v$, i.e., no message from $v$. By induction
  in round $i-1$ in $\G(v_0)$ node $v$ did not send a message neither
  to $w$ nor to $w'$ and $v'$ did not send a message to $w'$. Also, in
  round $i-1$ in $\G(v_0)$ node $z$ did send a message to $w$ or $w'$
  or $z'$ did send a message to $w'$. Thus, in $\G(v)$ in round $i$
  either $w$ sends a message to $v$ or $w'$ sends a message to $v$ or
  $v'$.

  Conversely suppose that one of the three events happens in $\G(v_0)$
  in round $i$. First assume that in round $i$ node $v$ received a
  message from $w$ in $\G(v_0)$. By Lemma~\ref{lem:simple} the message
  never left $\F(v_0)$, thus the message reached $w$ via a shortest
  path of $G$. This yields that in $G$, node $v$ also received in $G$
  message from $w$.

  Next suppose $v'$ receives a message from $w$ in $\G(v_0)$ in round
  $i$. Then $(v,w)$ is a cross edge of $G$. Thus,
  $d_{\G(v_0)}(v_0,w) = d_{\G(v_0)}(v_0,v)=i-1$ and hence,
  $d_G(v_0,w) = d_G(v_0,v)=i-1$. This implies that $v$ and $w$ do not
  send messages before round $i$ in $G$ and in round $i$ they send
  messages to each other in $G$, i.e., $v$ received in $G$ message
  from $w$.

  Finally suppose that $v'$ received a message from $w'$ in
  $\G(v_0)$ in round $i$. Lemma~\ref{lem:bilevels} gives
  $i=d_{\G(v_0)}(v_0,v')$ and $d_{\G(v_0)}(v_0,w')=i-1$. Let $P$ be a
  shortest path in $\G(v_0)$ from $v_0$ to $w'$. By
  Corollary~\ref{col:cross} $P$ uses a single cross edge $(x,y')$. Thus,
  in $G$ both nodes $x$ and $y$ receive a message in round
  $d_G(v_0,x)$, but not from each other. Hence, in round
  $d_G(v_0,x)+1$ both nodes $x$ and $y$ again receive a message. After
  another $d_G(y,w)$ rounds, $w$ sends a message to $v$. Thus, in $G$
  node $v$ receives in round
  $d_G(v_0,x)+1+d_G(y,v)= d_{\G(v_0)}(v_0,v')=i$ a message from
  $w$.
\end{proof}

The lemma proves that if no node in $G$ receives a message in a
specific round then no node in $\G(v_0)$ receives a message in this
round and vice versa. This yields the following theorem.
\begin{theorem}
  $\FL_G(v_0)=\FL_{\G(v_0)}(v_0)$ for every $v_0\in V$.\label{theo:aqui}
\end{theorem}

Note that if $G$ is bipartite then $\G(v_0)$ is disconnected and the
connected component of $\G(v_0)$ containing $v_0$ is just
$\F_{G(v_0)}$, i.e.,
$\epsilon_{\G(v_0)}(v_0)=\epsilon_{\F_{G(v_0)}}(v_0)=\epsilon_{G}(v_0)$
in this case. Lemma~\ref{lem:bipartite} and Theorem~\ref{theo:aqui}
together with Lemma~\ref{lem:equiv} imply the following result.
\begin{theorem}\label{theo:messages}
  Let $G(V,E)$ be a connected graph. Then
  $\FL_G(v_0)=\epsilon_{\G(v_0)}(v_0)$. Algorithm \SynFl sends $\abs{E}$
  messages if $G$ is bipartite and $2\abs{E}$ otherwise.
\end{theorem}%\todo{Maybe add: For each edge $(u,w)\in E$ we have $\abs{\FL_G(u)-\FL_G(w)}\le 1$.}
In case $G$ is non-bipartite for some edges messages are sent in both
directions, while for other edges two messages are sent in one
direction. With the introduced technique the proof of the main result
of \cite{Hussak:2020} becomes very simple.

%Next we give an upper bound for $\epsilon_{\G(v_0)}(v_0)$.

\begin{theorem}(Theorem 10 and 12, \cite{Hussak:2020})\label{theo:upper_bound}
  Let $G$ be a connected, non-bipartite graph and $v_0\in V$. Then
  $Rad(G) < \FL(v_0) \le \epsilon_{G}(v_0) + Diam(G) + 1$.
  Furthermore, $Rad(G) < \FL_1(G)\le Rad(G)+Diam(G)+1$.
  $Rad(G)=\FL_1(G)$ if and only if $G$ is bipartite.
\end{theorem}
\begin{proof}
  Let $u\in V$. Then $d_{\G(v_0)}(v_0,u) \le Diam(G)$. Thus, it
  suffices to give a bound for $d_{\G(v_0)}(v_0,u')$. Since $G$ is
  non-bipartite there exist cross edges with respect to $v_0$. Among
  all cross edges of $G$ choose $(v,w)$ such that
  $\min \{d_G(u,v),d_G(u,w)\}$ is minimal. WLOG assume that
  $d_G(u,v)\le d_G(u,w)$. Then the shortest path from $v$ to $u$ does
  not contain a cross edge (by choice of $(v,w)$. Thus, the distance
  from $v'$ to $u'$ in $\G(v_0)$ is at most $Diam(G)$.
  Hence,
  \[d_{\G(v_0)}(v_0,u') \le d_G(v_0,w) + 1 + Diam(G) \le
    \epsilon_{G}(v_0) + Diam(G) + 1.\] Hence,
  Theorem~\ref{theo:messages} implies the upper bound. Let $v$ be a
  node with $dist_G(v_0,v)\ge Rad(G)$. Then
  $d_{\G(v_0)}(v_0,v') \ge Rad(G)+1$. This yields the lower bound.
  Since this is true for all $v_0\in V$ the second statement also
  holds. Now the last statement follows from
  Corollary~\ref{cor:bipartite}.
\end{proof}

The above upper bound is sharp as can be seen for $G=C_n$ with
$n\equiv 1 (2)$. In this case $Rad(C_n)=Diam(C_n)=(n-1)/2$ and
$\FL_1(C_n)=n$. Fig.~\ref{fig:bipartite} shows on the left a
non-bipartite graph with $Rad(G)+1=\FL_1(G)$.

\section[The Case |S|>1]{The Case $|S|>1$}
\label{sec:several-sources}
The case $|S|>1$ requires a slightly different definition of the
auxiliary graph $\G$. First, a virtual source $v^\ast$ connected by
edges to all source nodes in $S$ is introduced. Call this graph
$G^\ast(S)$. Note that even in case $G$ is bipartite $G^\ast(S)$ is
not necessarily bipartite. Fig.~\ref{fig:aux} shows an example for
$G^\ast(S)$. When $v^\ast$ sends a message to all neighbors, then in
the next round all nodes of $S$ send a message to all their neighbors
except $v^\ast$. Thus, the initial behavior of \SynFl is the same for
$G$ with a given set $S$ and for $G^\ast(S)$ when started by node
$v^\ast$ only. Later the behavior may deviate because nodes of $S$ may
send a message to $v^\ast$. To eliminate this effect we change the
definition of the graph $\G(S)$ as follows. The auxiliary graph
$\G(S)$ is the graph $\G(v^\ast)$ constructed from $G^\ast(S)$ as in
section~\ref{sec:subst-graph} with the only difference that the copy
of $v^{\ast}$ in the second copy of $\F(S)$ is removed. More formally.

\begin{figure}%
  \hfill
    \includegraphics[valign=t]{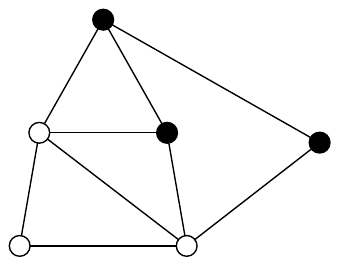}
  \hfill
    \includegraphics[valign=t]{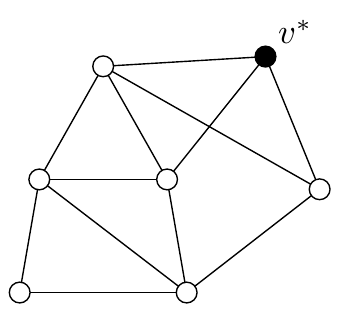}
  \hfill\null
  \caption{On the left a graph with $|S|=3$; on the right the auxiliary graph
    $G^\ast(S)$. \label{fig:aux}}
\end{figure}

\begin{definition}
  Let $S\subseteq V$. Denote by $\F_S(v^\ast)$ the subgraph of
  $G^\ast(S)$ with node set $V\cup \{v^\ast\}$ and all edges of
  $G^\ast(S)$ that are not cross edges with respect to $v^\ast$. Let
  $\F_S'(v^\ast)$ be a copy of $\F_S(v^\ast)$ without the node
  $v^\ast$ and the incident edges. Denote by $V'$ the node set of
  $\F_S'(v^\ast)$. Let $\G(S)$ the graph with node set
  $(V\cup \{v^\ast\})\cup V'$ that consists of $\F_S(v^\ast)$ and
  $\F_S'(v^\ast)$ and additionally for any cross edge $(u,w)$ of $G$ the
  edges $(u,w')$ and $(w,u')$.
\end{definition}

Fig.~\ref{fig:aux_graph} shows $\G(S)$ where graph $G$ and set $S$ are
taken from Fig.~\ref{fig:aux}. $\G(S)$ has $2\abs{E} + \abs{S}$ edges
and $2\abs{V}+1$ nodes.

%Note that Theorems~\ref{theo:aqui} and \ref{theo:messages} still hold.

\begin{figure}[h]%
  \hfill
    \includegraphics[valign=t]{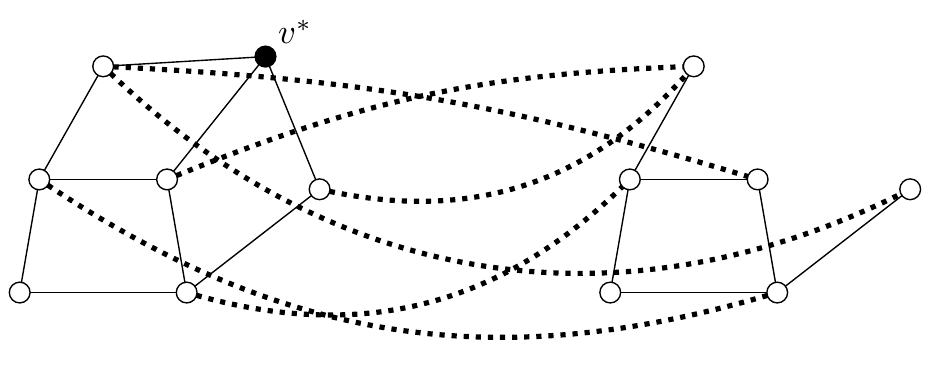}
  \hfill\null
  \caption{The graph $\G(S)$ for the graph $G^\ast(S)$ depicted in
    Fig.~\ref{fig:aux}. Edges connecting $\F_S(v^\ast)$ and
    $\F'_S(v^\ast)$ are displayed as dotted lines.\label{fig:aux_graph}}
\end{figure}

% \begin{theorem}\label{theo:fundametal}
%   Let $G(V,E)$ be a connected graph and $S\subseteq V$. Then
%   $\FL_G(S)+1=\FL_{\G(S)}(v^\ast)$. Algorithm \SynFl sends $\abs{E}$
%   messages if $G^\ast(S)$ is bipartite and $2\abs{E}$ otherwise.
% \end{theorem}

\subsection{Proof of Theorem~\ref{theo:fundametal} and
  Theorem~\ref{theo:mainRes}}
\label{sec:proof1_and_2}

\begin{proof}(Theorem~\ref{theo:fundametal}) Clearly $\G(S)$ is
  bipartite. We prove by induction on $i$ that the behavior of \SynFl
  in round $i$ on graph $\G(S)$ when started by the node $v^\ast$ is
  the same as that of \SynFl in round $i+1$ on $G$ when started by the
  nodes in $S$. The statement is obviously true for $i=1$. Let $i>1$.
  Note that $\FL_G(S)$ is not in all cases equal to
  $\FL_{G^\ast(S)}(v^\ast)-1$. In $G^\ast(S)$ a node in $S$ may send
  in round $i>1$ a message to $v^\ast$ which will cause further
  messages between nodes in $G^\ast(S)$ which have no counterpart in
  $G$. This can only happen if a message has passed a cross edge of
  $G^\ast(S)$; consider a graph with a single edge where both nodes
  belong to $S$. In $\G(S)$ the situation is different. Messages sent
  via cross edges in $G$ are in $\G(S)$ sent into $\F_S'(v^\ast)$ and
  by Lemma~\ref{lem:simple} they will never leave that component.
  Since $\F_S'(v^\ast)$ doesn't contain a copy of $v^\ast$ the message
  flow will terminate as in $G$. Hence, we can repeat the proof of
  Lemma~\ref{lem:equiv} to show $\FL_G(S) = \FL_{\G(S)}(v^\ast)-1$.
  Since $\G(S)$ is bipartite Corollary~\ref{cor:bipartite} implies
  $\FL_{\G(S)}(v^\ast) = \epsilon_{\G(S)}(v^\ast)$. Now the
  construction of $\G(S)$ implies
  $\FL_{\G(S)}(v^\ast) \le 1 +d_G(S,V) +1 + Diam(G)$. Hence,
  $\FL_G(S) \le d_G(S,V) +1 + Diam(G)$. The statement about the number
  of messages follows immediately from the structure of $\G(S)$: One
  message is sent over each edge of $\F_S(v^\ast)$ and $\F_S'(v^\ast)$
  and two messages are sent via each cross edge of $G$.
\end{proof}

\begin{lemma}\label{lem:tree_gamma}
Let  $T$ be a tree and $k\ge n/2$. Then $r^{ni}_k(T)\le 2$.
\end{lemma}
\begin{proof}
  The result is clear for $n <5$. Let $s$ be the number of leaves of
  $T$. Let $T'$ be the tree obtained from $T$ by removing all leaves.
  Denote by $s'$ the nodes of $T'$ that have a neighbor of degree $1$.
  By Theorem 7(a) of \cite{Chellali:2004} we $\gamma_t(T')\le (n-s+s')/2$.
  Since the number of leaves of $T'$ is less then the number of leaves
  of $T$ we have $s'\le s$. This yields $\gamma_t(T')\le n/2$. Hence,
  $r^{ni}_k(T')=1$ for $k\ge n/2$ and therefore $r^{ni}_k(T)\le 2$. 
\end{proof}

This bound is sharp. Fig.~\ref{fig:kamm} shows a tree $T$ with
$r^{ni}_4(T)= 2$, the black nodes form a non-isolated $4$-center of $T$
of minimal radius. Also, $r^{ni}_7(T)= 2$.

\begin{figure}[h]%
  \hfill
  \includegraphics{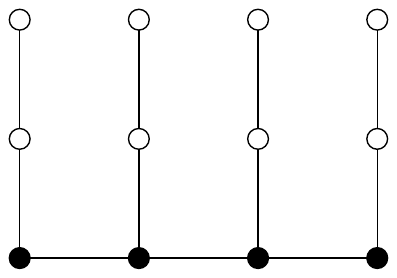}
  \hfill\null
  \caption{A bipartite graph $G$ with $\FL_7(G) -r_7(G)=2$.
    Furthermore, $r_4(G)+1 = \FL_4(G)=2$,
    $\FL_7(G)= r^{ni}_7(G)+1= 3$, $\FL_8(G)= r_4(G)+1= 2$, and
    $\FL_8(G)= Rad(G)+1= 2$.\label{fig:kamm}}
\end{figure}

\begin{lemma}\label{lem:bipar_gamma}
  Let $G$ be a connected graph and $k\ge n/2$. Then
  $r^{ni}_k(G)\le 2$.
\end{lemma}
\begin{proof}
  The proof is by induction on $m$, the number of edges of $G$. If
  $m=n-1$ then $G$ is a tree and Lemma~\ref{lem:tree_gamma} yields
  the result. Let $m> n-1$. Consider an edge $e$ such that
  $G \setminus e$ is connected. By induction
  $r^{ni}_k(G\setminus e)\le 2$. Since adding an edge does not
  increase the non-isolated k-radius, the proof is complete.
\end{proof}

\begin{proof}(Theorem~\ref{theo:mainRes})%  qqq
  (1) Obviously, $r_k(G)\le \FL_k(G)$. Let $S$ be a non-isolated
  $k$-center of $G$ and $v\in V$. Then there exists $u\in S$ such that
  $d_G(u,v)\le r^{ni}_k(G)$ and the path from $u$ to $v$ in
  $G^\ast(S)$ does not use a cross edge with respect to $v^\ast$.
  Since $S$ contains no isolated node, after the first round each node
  of $V'$ that is a copy of a node in $S$ receives a message. Thus,
  there exists $w\in N(u) \cap S$. Hence, the path $v^\ast, w,u'$
  exists in $\G(S)$. Therefore, the distance from $v^\ast$ to $v'$ in
  $\G(S)$ is at most $2 + dist_G(u,v)$, i.e.,
  $\epsilon_{\G(S)}(v^\ast)\le 2 + r^{ni}_k(G)$. Hence by
  Lemma~\ref{lem:bipartite} and Corollary~\ref{cor:bipartite} we have
  $\FL_{\G(S)}(v^\ast)\le 2 + r^{ni}_k(G)$. The first bound follows
  from Theorem~\ref{theo:fundametal}.

  Let $S\subset V$ with $|S|=\lfloor k/2\rfloor$ and
    $d_G(S,V) = r_{\lfloor k/2\rfloor}(G)$. Obviously, there exists a
    subset $S'$ of $V$ such that $|\hat{S} \cup S'|=k$ and
    $G[\hat{S} \cup S']$ has no isolated node. Let
    $\hat{S}=S \cup S'$. Then,
    $r^{ni}_k(G) \le dist(\hat{S},V)\le r_{\lfloor k/2\rfloor}(G)$.

  (2) Let $\FL_k(G)=1$. Then each node $v\in V$ must have all its
  neighbors in $S$ or none. Let $S_1$ be the set of nodes that have
  all their neighbors in $S$. Let $v\in V\setminus S_1$. Then
  $N(v)\cap S = \emptyset$. Hence, $v$ does not receive a message
  in the first round. Since $\FL_k(G)=1$ and $G$ is connected, we have
  $v\in S$. Thus $V\setminus S_1 \subseteq S$.

  Assume there exists $v\in S_1$ with $N(v)\subseteq S_1$. Then
  $v\in S$. This yields $N(N(v))\subseteq S_1$ and consequently
  $V=S_1$ and thus $V=S$ since $G$ is connected, i.e., $n=k$. Next
  assume $N(v)\not \subseteq S_1$ for all $v\in S_1$. Thus, for
  $v_1\in S_1$ there exits a neighbor $v_2$ that is not in $S_1$,
  i.e., $N(v_2)\cap S = \emptyset$. Then $v_2\in S$ because
  $N(v_1)\subseteq S$. If $v_1$ would be in $S$, then all neighbors of
  $v_2$ would be in $S$ and thus, $v_2\in S_1$. Thus, $v_1\notin S$.
  Hence, $S_1\cap S=\emptyset$. Therefore, $S$ and $S_1$ are
  independent sets. Also $S_1 \cup S=V$. Thus, $G$ is bipartite. Since
  the opposite direction is trivially true, the proof is complete.

  (3) Let $S$ be a minimal total dominating set of $G$. Then $G[S]$ has
  no isolated node and $d(S,V)=1$. Hence, for $k\ge \gamma_t(G)$ we
  have $r^{ni}_k(G)=1$. For $k > \gamma_t(G)$ we add $k-\gamma_t(G)$
  nodes to $S$ such that the induced subgraph of the resulting set $T$
  has no isolated nodes. Obviously, we still have $d(T,V)=1$. Thus,
  $\FL_{\G(S)}(v^\ast)\le 3$. Now, Theorem~\ref{theo:fundametal} yields
  $\FL_{G}(T)\le2$.

  There are several upper bounds known for
  $\gamma_t(G)$. Cockayne et al.\ proved that $\gamma_t(G)\le
  2n/3$ for every connected graph $G$ with $n\ge
  3$ \cite{Cockayne:1980}. Archdeacon et al.\ improved this result
  such that if $\delta(G)\ge 3$ then $\gamma_t(G)\le
  n/2$ \cite{Archdeacon:2004}. These results complete the proof.

  (4) If $k\ge n/2$ then $r^{ni}_k(G)\le 2$ by
  Lemma~\ref{lem:bipar_gamma}. Thus,
  Theorem~\ref{theo:mainRes} (\ref{theo:mainRes1}) implies the
  bound.
  \end{proof}

  % Let $S$ be a minmal total dominating set. Then $S$
  %   is a non-isolated $\gamma_t(G)$-center of $G$ with non-isolated
  %   $\gamma_t(G)$-radius $1$. For $k > \gamma_t(G)$ we add
  %   $k-\gamma_t(G)$ nodes to $S$ such that the induced subgraph of the
  %   resulting set $T$ has no isolated nodes. Obviously, we still have
  %   $d(T,V)=1$. Thus, $\FL_{\G(S)}(v^\ast)=3$. Now,
  %   Theorem~\ref{theo:fundametal} yields $\FL_{G}(T)=2$. Hence,
  %   $\FL_k(G)=2$ since $\FL_k(G)>1$.
%  Der dritte Teil wird später bewiesen, der dritte Beweisteil
%  beweist Theorem 3.1
% \begin{corollary}\label{cor:ni}
%   $\FL_k(G)\le r^{ni}_k(G) +1$ for each connected graph $G$ and
%   $k> 1$.
% \end{corollary}
% \begin{proof}
%   Let $S$ be a non-isolated $k$-center of $G$ and $v\in V$. Then there
%   exists $u\in S$ such that $d_G(u,v)\le r^{ni}_k(G)$ and the path
%   from $u$ to $v$ in $G^\ast(U)$ does not use a cross edge with
%   respect to $v^\ast$. Also there exists $w\in N(u) \cap U$. Hence,
%   the path $v^\ast, w,u'$ exists in $\G(U)$. Therefore, the distance
%   from $v^\ast$ to $v'$ in $\G(U)$ is at most $2 + dist_G(u,v)$. The
%   result follows from Theorem~\ref{theo:fundametal}.
% \end{proof}

  The first bound of Theorem~\ref{theo:mainRes} (\ref{theo:mainRes1})
  is sharp. For $C_n$ with $n\equiv 3 (4)$ we have
  $(n-3)/4 +1 =\FL_4(C_n)=r^{ni}_4(C_n)+1$. On the other hand
  $r^{ni}_3(C_n)-\FL_3(C_n)$ and $\FL_3(C_n) -r_3(C_n)$ with
  $n\equiv 1 (2)$ are unbounded for growing $n$. The example in
  Fig.~\ref{fig:kamm} shows that all bounds of
  Theorem~\ref{theo:mainRes} (\ref{theo:mainRes1}) are sharp.

Note that $r_k(G)=1$ for $k\ge n/2$ \cite{Meir:1975}. Hence,
Theorem~\ref{theo:mainRes} (\ref{theo:mainRes4}) yields
$\FL_k(G)-r_k(G) \le 2$ if $k\ge n/2$. Dankelmann et al.\ provide
several upper bounds for $r_k(G)$ in terms of $n$ and $\delta$
\cite{Dankelmann:2019}. These can be used to state bounds for
$\FL_k(G)$ in terms of $n$ and $\delta$. The following result is a
consequence of Theorem 14 of \cite{Dankelmann:2019} and
Theorem~\ref{theo:mainRes} (\ref{theo:mainRes1}).
\begin{corollary}
  Let $G$ be a connected, triangle-free non-bipartite graph and $1<k<n$. Then
  \[\FL_k(G)\le \frac{2(n-1)}{\delta(\lfloor
    k/2\rfloor +1)}+5.\]
\end{corollary}

\subsection{Proof of Theorem~\ref{theo:nonbip} }
\label{sec:proof3}
If $G$ is non-bipartite, then $G^\ast(S)$ is also non-bipartite. Then
$\FL_k(G)\ge r_k(G)+1$ by Theorem~\ref{theo:fundametal}. Hence,
Theorem~\ref{theo:nonbip} (\ref{theo:nonbip1}) holds. In this section
we prove a lower bound for $\FL_k(G)$ that depends only on $Rad(G)$
and $k$.

\begin{lemma}\label{lem:treeradius1}
  For $k>0$ and all trees $T$ we have  $kr_k(T)\ge Rad(T)-k/2$.
\end{lemma}
\begin{proof}
  Let $v,w\in V$ such that $d_G(v,w)=Diam(T)$. In the best case two
  consecutive nodes on the path from $v$ to $w$ that belong to a
  $k$-center have distance $2r_k(T)+1$. Thus,
  $2kr_k + k -1 \ge Diam(T)$. This yields $2kr_k + k -1 \ge 2Rad(T)-1$
  which proves the result.
\end{proof}

Note that the bound of this lemma is sharp. Consider a path $P$ with
$k(2c+1)$ nodes for $c>0$. Then $r_k(P)=c$ and $Rad(P)=k(2c+1)/2$.

%The proof of the next Lemma can be found in the appendix.

\begin{lemma}\label{lem:span_tree}
  Let $G$ be a connected graph. Then there exists a spanning tree $T$
  of $G$ such that $r_k(G)=r_k(T)$ for all $k\ge 1$.\label{lem:tree_graph}
\end{lemma}
\begin{proof}
  Let $U=\{u_1,\ldots,u_k\}$ be a $k$-center of $G$. Let
  $V_1=\{v\in V\suchthat d_G(v,u_1)\le d_G(v,u_j) \text{ for }
  j=2,\ldots,k \}$. For $i=2, \ldots, k$ let
  $V_i = \{v\in V\setminus (V_1\cup\ldots \cup V_{i-1})\suchthat
  d_G(v,u_i)\le d_G(v,u_j) \text{ for } j=i+1,\ldots,k \}$. Note that
  $u_i\in V_i$. Clearly, the $V_i$ form a partitioning of $V$. Let
  $v \in V_i$ and $u$ be a node on the shortest path from $u_i$ to $v$.
  Thus, $d_G(u_i,v)\le d_G(u_j,v)$. Assume that $u\in V_j$. Then
  $d_G(u_j,u)\le d_G(u_i,u)$. This implies that $d_G(u_j,u)=d_G(u_i,u)$
  and $d_G(u_j,v)=d_G(u_i,v)$. Since $v\in V_i$ we have that $i\le j$ and
  hence $u \in V_i$. Thus, $u\in V_i$ and hence $G_i=G[V_i]$ is a
  connected graph. Let $T_i$ be a breadth-first-search tree of $G_i$ rooted
  in $u_i$. By adding $k-1$ edges the $T_i$'s can be combined into a
  spanning tree $T$ of $G$. Let $w_i$ be a central node of $T_i$ and
  $U_T=\{w_1,\ldots,w_k\}$. For each $j =1,\ldots, k$ we have
  \[Rad(T_j) = \max_{v\in V_j} d_{T_j}(v,w_j) \le \max_{v\in V_j}
    d_{T_j}(v,u_j)= \max_{v\in V_j} d_{G}(v,u_j)\le r_k(G).\] This
  yields
  \[r_k(T)\le d_T(U_T,V) \le \max_{j=1,\ldots,k} Rad(T_j) \le r_k(G)
    \le r_k(T).\]This completes the proof.
\end{proof}

The graph shown in Fig.~\ref{fig:kamm} demonstrates that $\FL_k(G)$ is
not always decreasing for increasing $k$:
$\FL_5(G)=1,\FL_6(G)=3,\FL_7(G)=2,\FL_8(G)=2$. Next we show that for
non-bipartite graphs these values monotonically decrease.

\begin{lemma}\label{lem:monoton}
  Let $G$ be a connected non-bipartite graph,
  $\emptyset \not=S\subset V$, and $v\in V\setminus S$. Then
  $\FL_G(S\cup \{v\}) \le \FL_G(S)$.
  \end{lemma}
  \begin{proof}%[Lemma~\ref{lem:monoton}]
    Let $\hat{S}=S\cup \{v\}$. It suffices to prove that
    $d_{\G(\hat{S})}(v^\ast,u)\le d_{\G({S})}(v^\ast,u)$ for all
    $u\in V\cup V'$. Obviously
    $d_{G^\ast(\hat{S})}(v^\ast,w) \le d_{G^\ast(S)}(v^\ast,w)$ for
    all $w\in V$. Hence, we only have to consider $u\in V'$, i.e.,
    there exists $w\in V$ with $w'=u$.

    Let ${\cal P}_S(w)$ be the set of all path in $G^\ast(S)$ from
    $v^\ast$ to $w$ that use exactly one cross-edge of $G^\ast(S)$.
    Let $P \in {\cal P}_S(w)$ be a path of shortest length among all
    path in ${\cal P}_S(w)$. By Corollary~\ref{col:cross} the length of
    $P$ is equal to $d_{\G({S})}(v^\ast,w)$. It suffices to prove that
    there exists a path $P'\in {\cal P}_{\hat{S}}(w)$ that has at most
    the length of $P$. Consider $P$ as a path in $G^\ast(\hat{S})$.
    Then some cross edges may have become forward edges and vice
    versa. Let $s$ be the number of cross edges of $P$. Assume $s=0$.
    Then $d_{G^\ast(\hat{S})}(v^\ast,w)$ is the length of $P$. But the
    length of $P$ is also equal to $d_{G^\ast(\hat{S})}(v^\ast,w)+1$.
    Thus, $d_{G^\ast(\hat{S})}(v^\ast,w) > d_{G^\ast(S)}(v^\ast,w)$.
    Contradiction and thus, $s>0$. Assume that $s>1$. Let $(v_1,v_2)$
    be the first cross edge of $P$ in $G^\ast(\hat{S})$. Then there
    exists a path $P_1$ in $G^\ast(\hat{S})$ from $v^\ast$ to $v_2$
    that consist s of forward edges only. Construct a new path $P'$ by
    appending to $P_1$ the subpath of $P$ from $v_2$ to $w$. Then $P'$
    is shorter than $P$ and consists of $s-1$ cross edges. Repeating
    this construction will lead to a path
    $Q \in {\cal P}_{\hat{S}}(w)$ that is no longer than $P$. 
  \end{proof}

  The reason that Lemma~\ref{lem:monoton} does not hold for a
  bipartite graph $G$ is that depending on $S$ and $v$ the graph
  $G^\ast(S)$ may be bipartite while $G^\ast(S\cup \{v\})$ is
  non-bipartite.

\begin{proof}(Theorem~\ref{theo:nonbip})%  xxx
  (2) If $r^{ni}_k(G)=1$ then $\FL_k(G)\le 2$ by
  Theorem~\ref{theo:mainRes} (\ref{theo:mainRes1}) and thus $\FL_k(G)=2$
  by the first part. Next suppose that $\FL_k(G)= 2$ and let $S$ be a
  subset of $V$ such that $\FL_G(S)=2$. Assume that $G[S]$ contains an
  isolated node $v$. Since $G$ is non-bipartite the shortest path in
  $\G(S)$ from $v^\ast$ to $v'$ has length at least $4$. Then
  Theorem~\ref{theo:fundametal} implies that
  $\FL_k(G)\ge 3$. This contradicts the assumption that $\FL_k(G)= 2$.
  Therefore, $G[S]$ contains no isolated node and hence
  $r_k(G)=r^{ni}_k(G)$. Since $2 = \FL_k(G) \ge r_k(G)+1$ we have
  $r^{ni}_k(G)=1$.

  (3) By Lemma~\ref{lem:tree_graph} there exists a spanning tree $T$
  of $G$ such that $r_k(G)=r_k(T)$. Note that
  $r_k(G) +1 \le \FL_k(G)$. By Lemma~\ref{lem:treeradius1}
  $kr_k(T) \ge Rad(T)-k/2$. Hence
  \[k\FL_k(G) \ge kr_k(G) +k = kr_k(T) +k\ge Rad(T) + k/2\] since
  $Rad(T)\ge Rad(G)$.

      Let $S\subset V$ with $|S|=k$ such that $\FL(S)=\FL_k(G)$. Let
    $v\in V\setminus S$. Then Lemma~\ref{lem:monoton}
    yields $\FL_{k+1}(G) \le \FL(S\cup\{v\}) \le \FL(S) = \FL_k(G)$.
\end{proof}

Theorem~\ref{theo:nonbip}(\ref{theo:nonbip2}) does not hold for
bipartite graphs as a path $P$ with $12$ nodes demonstrates,
$\FL_3(P)=2$ and $r^{ni}_3(P)=5$.
Theorem~\ref{theo:mainRes} (\ref{theo:part2}) implies that
$r_k(G)\in \{1,2\}$ if $G$ is bipartite and $\FL_k(G)=2$. Thus,
$r_k(G)=1$ implies $\FL_k(G)\le2$. Bipartite graphs with $r_k(G)=2$
can have $\FL_k(G)> 2$ as the example in Figure~\ref{fig:bipartite}
shows.

\subsection{Proof of Theorem~\ref{theo:special} }  %yyyy
\label{sec:proof4}

For $k=1$ we have $\FL_1(G)=r_1(G)$ provided $G$ is bipartite and vice
versa (Theorem~\ref{theo:upper_bound}, Theorem 11 \cite{Hussak:2020}).
For $k>1$ we have a slightly different situation.

\begin{proof}(Theorem~\ref{theo:special})
  (1) If $G$ has a $k$-center $S$ that is contained in $V_1$ or $V_2$ then
  the graph $G^\ast(S)$ has no cross edge with respect to $v^\ast$
  (nodes with the same distance to $v^\ast$ are either in $V_1$ or
  $V_2$). Thus, Theorem~\ref{theo:fundametal} implies
  $\FL_k(G)=r_k(G)$.

  Next assume that $\FL_k(G)=r_k(G)$. Let $S\subset V$ with $|S|=k$
  and $\FL(S)=\FL_k(G)$. Then $r_k(G)\le d_G(S,V)\le \FL(S) =r_k(G)$,
  i.e., $d_G(S,V)=r_k(G)$. Since $\FL_k(G)=r_k(G)$, $G^\ast(S)$ does
  not contain a cross edge with respect to $v^\ast$. Let
  $S_i=S\cap V_i$. Denote by $V^i$ the set of nodes that receive a
  message when \SynFl is executed for set $S_i$ (including $S_i$).
  Since there are no cross edges, there exists no edge connecting a
  node from $V^1$ with a node from $V^2$. Since $G$ is connected,
  either $S_1=\emptyset$ or $S_2=\emptyset$. This implies the result.

  (2)   Let $S\subseteq V=V_1\cup V_2, E$ with $|S|=k$ such that
  $d_G(S,V)=r_k(G)$. If $S\subseteq V_1$ or $S\subseteq V_2$ then the
  result follows from the first part. WLOG assume that
  $k\le \abs{V_2}$. Let $S_1 = S \cap V_1$. There exists
  $T\subseteq V_2$ with $|T|\le |S_1|$ such that
  $S_1 \subseteq N(T)\subseteq V_1$. Then
  $d_G(T\cup (S \cap V_2), V) \le r_k(G) + 1$. Let $T_r\subseteq V_2$
  such if $S_n=T_r\cup T\cup (S \cap V_2)$ then $|S_n|=k$. Clearly,
  $d_G(S_n,V)\le r_k(G)+1$. Obviously, $G^\ast(S_n)$ does not contain
  a cross edge, since the graph $G^\ast(S_n)$ is bipartite. Hence,
  $\FL_k(G) -r_k(G)\le 1$ by Theorem~\ref{theo:fundametal} and
  Corollary~\ref{cor:bipartite}.

  (3) If $k\ge n/2$ then $r_k(G)=1$. Thus, $\FL_k(G) -r_k(G)\le 2$ by
  Theorem~\ref{theo:mainRes} (\ref{theo:mainRes4}). The case $k < n/2$
  is implied by (2).  
\end{proof}

Fig.~\ref{fig:bipartite} shows on the right a bipartite graph $G$ with
$r_2(G)=2$ and $\FL_2(G)=3$. Thus, the bound of
Theorem~\ref{theo:special} (\ref{theo:special2}) is sharp.
Fig.~\ref{fig:kamm} demonstrates that the bound of
Theorem~\ref{theo:special} (\ref{theo:special3}) is sharp.

  % by \cite{Meir:1975}
\begin{figure}%
  \begin{minipage}{0.49\textwidth}
  \hfill
    \includegraphics{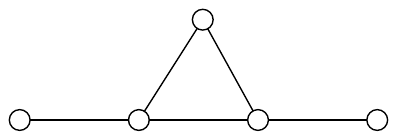}
  \hfill\null
  \end{minipage}
  \hfill
  \begin{minipage}{0.49\textwidth}
  \hfill
    \includegraphics{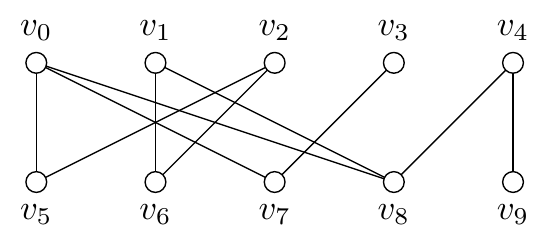}
  \hfill\null
  \end{minipage}
  \caption{Left: A non-bipartite graph with $Rad(G)=2$ and
    $\FL_1(G) = 3$. Right: A bipartite graph with $r_2(G)=2$ and
    $\FL_2(G)=3$ ($d(\{v_0,v_8\},V) =2$,
    $\FL(\{v_0,v_1\})=3$). \label{fig:bipartite}}
\end{figure}

\subsection{Proof of Theorem~\ref{theo:NP} } 
\label{sec:proof5}

\begin{proof}(Theorem~\ref{theo:NP})
  Clearly the $(k,c)$-flooding problem is in $NP$. We will reduce the
  $NP$-complete problem of total domination to $(k,c)$-flooding
  \cite{Henning:2009}. Let $\cal A$ be an algorithm that solves the
  $(k,c)$-flooding problem in polynomial time. Let $G(V,E)$ be a graph. If
  $G$ is bipartite then we select a node $v\in V$ and attach the
  gadget shown in Fig.~\ref{fig:gadget}. The resulting graph $G'$ is
  obviously non-bipartite. First we prove that $G$ has a total
  dominating set of cardinality at most $k$ if and only if $G'$ has a
  total dominating set of cardinality at most $k+ 2$.
  Let $S$ be a total dominating set of $G$. Then $S\cup \{b,c\}$ is a
  total dominating set of $G'$ with $\abs{S}+2$ nodes. Let $T$ be a
  total dominating set of $G'$. Then $c\in T$. If $b\not\in T$ then
  $d\in T$ or $e\in T$. Thus, $T\setminus \{d,e\} \cup \{b\}$ is
  also a total dominating set of $G'$. Hence, we can assume that
  $b\in T$. If $a\in T$ then $T\setminus \{a\} \cup \{v\}$ is also
  a total dominating set of $G'$. Hence, we can assume that
  $a\not \in T$. Therefore, $T\setminus \{a,b,c,d,e\}$ is a total
  dominating set of $G$ with at most $\abs{T} - 2$ nodes.

  Thus, we can assume $G$ is non-bipartite. Next we apply algorithm
  $\cal A$ to $G$. By Theorem~\ref{theo:nonbip} (\ref{theo:nonbip2}) we
  have $\FL_k(G)= 2$ if and only if $r^{ni}_k(G)=1$. Note that
  $r^{ni}_k(G)=1$ implies that $G$ has a total domination set of order
  at most $k$ and vice versa. Hence, algorithm $\cal A$ can decide whether $G$
  has a total domination set of size $k$.
\end{proof}

\begin{figure}[h]%
  \hfill
  \includegraphics{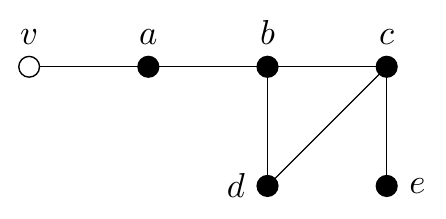}%{gadget}
  \hfill\null
  \caption{This gadget consisting of five nodes is attached to a node $v$.\label{fig:gadget}}
\end{figure}

\section{Conclusion and Future Work}%\todo{Überarbeiten!}
In this paper we analyzed amnesiac flooding for a set $S$ of $k$
initiators and introduced the $(k,c)$-flooding problem. The main technical
result is the construction of a bipartite graph $\G(S)$ such that the
executions of amnesiac flooding on $G$ and $\G(S)$ are equivalent.
This allowed us to prove upper and lower bounds for the round
complexity of amnesiac flooding. Furthermore, we showed the
relationship between the $k$-center and $(k,c)$-flooding for bipartite
graphs and proved that $(k,c)$-flooding is $NP$-complete.

There are several open problems related to amnesiac flooding. Firstly,
we suspect that the first bound stated in
Theorem~\ref{theo:nonbip} (\ref{theo:part3}) is not sharp. Instead we
have the following conjecture: If $G$ is connected, non-bipartite then
$k\FL_k(G)\ge Rad(G) + k -1$. If $\FL_k(G) \ge r_k(G) +2$ then the
proof of Theorem~\ref{theo:nonbip} (\ref{theo:part3}) can be used to
prove this conjecture. Thus, in proving the conjecture one can assume
$\FL_k(G) = r_k(G) +1$. This new bound would be sharp. Let $H_{12}$ be
the graph with 12 nodes as depicted in Fig.~\ref{fig:h12}. Connect
eight copies of $H_{12}$ by adding $7$ edges connecting the copies one
after the other at the end nodes. The resulting graph $G$ has $96$
nodes, $Rad(G)=40$, and $\FL_3(G)=14$.

Secondly, by Theorem~\ref{theo:special} (\ref{theo:special3})
$\FL_k(G)$ assumes one of three values in case $G$ is bipartite. Is it
possible to infer from structural parameters of $G$ the value of
$\FL_k(G)$ in this case? A naive approach to determine $\FL_k(G)$
requires $O(n^km)$ time.

\begin{figure}[h]%
  \hfill
    \includegraphics[valign=t]{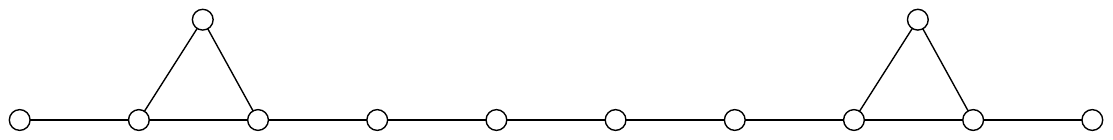}
  \hfill\null
  \caption{Eight copies of this graph glued together prove the
    sharpness of the lower bound
    $k\FL_k(G)\ge Rad(G) + k -1$.\label{fig:h12}}
\end{figure}

Theorem~\ref{theo:mainRes} (\ref{theo:mainRes1}) suggests that the
well-known greedy algorithm with approximation ratio $2$ for the
metric $k$-center might be a good heuristic to determine a set $S$
with small value of $\FL_k(G)$. The design of an approximation
algorithm for the $(k,c)$-flooding problem with approximation ratio
$2$ is an open problem. Also, distributed algorithms for this problem
have not been researched.

% Unfortunately, the bound can be arbitrarily bad as the graphs $C_n$
% with $n\equiv 1(2)$ show. For $k=3$ the sequence $\FL_3(C_n) -r_3(C_n)$
% is unbounded. 
% re exists a simple greedy algorithm with
% approximation ratio $2$ for the metric $k$-center problem. Is there a
% similar approximation algorithm for the $k$-flooding problem?

Another open issue is to devise a stateless information dissemination
algorithm for synchronous distributed systems with a better round
complexity than \SynFl. Is it possible to achieve the optimal time
complexity of $Diam(G)$ rounds? Algorithm \SynFl cannot be executed in
an asynchronous system. Hussak et al.\ showed that a simple adaptation
of amnesiac flooding to asynchronous systems does not terminate
\cite{Hussak:2020}. Does a stateless asynchronous information
dissemination algorithm exist?

Denote by $d_r(v,w)$ the number of the round in which node $w$
receives the last message when amnesiac flooding is started in node
$v$. It is straightforward to prove that for non-bipartite graphs this
does not define a metric but a meta-metric in the sense
of~\cite{Vaisala:2005}. Hence it it can be used to quantify the
importance of a node in a given network, i.e., it defines a centrality
index \cite{Zweig:2016}. There are many centrality indices proposed in
the literature (degree, closeness, betweenness, eigenvector centrality
etc.). The question is whether it coincides with any of the known
centrality indices.

\bibliographystyle{unsrt}
\bibliography{document}

\end{document}